\documentclass[pre,twocolumn,showpacs,superscriptaddress,preprintnumbers,floatfix]{revtex4-1}

\usepackage{etex}
\usepackage{ifpdf}
\usepackage{hyperref}
\usepackage{dcolumn}
\usepackage{url}
\usepackage{amsmath}
\usepackage{amscd}
\usepackage{amsfonts}
\usepackage{amssymb}
\usepackage{bm}   
\usepackage{bbm}
\usepackage{verbatim}
\usepackage{stmaryrd}
\usepackage{amsthm}
\usepackage{xcolor}
\usepackage{setspace}

\ifpdf
  \usepackage[pdftex]{graphicx}
\else
  \usepackage[dvips]{graphicx}
\fi

\newif \ifgenerate
\generatefalse 

\def\baseimagedir{}
\def\tikzpath{\baseimagedir tikz/fragments/}

\ifgenerate
  \usepackage{tikz}
  \input{\tikzpath tikzlibraries}
  \input{\tikzpath vaucanson.tikz}
\fi

\theoremstyle{plain}	
\theoremstyle{plain} 	
\theoremstyle{plain} 	\newtheorem{The}{Theorem}
\theoremstyle{plain} 	
\theoremstyle{plain} 	
\theoremstyle{plain}	
\theoremstyle{plain}	\newtheorem{Def}{Definition}
\theoremstyle{plain}	

\makeatletter
\renewenvironment{proof}[1][\proofname]{%
\par\pushQED{\qed}\normalfont%
\topsep6\p@\@plus6\p@\relax\trivlist%
\item[\hskip\labelsep\bfseries#1\@addpunct{.}]\itshape\ignorespaces}{%
\popQED\endtrivlist\@endpefalse}%
\makeatother

\def\clap#1{\hbox to 0pt{\hss#1\hss}}

\newcommand{\eM}     {\mbox{$\epsilon$-machine}}
\newcommand{\eMs}    {\mbox{$\epsilon$-machines}}

\newcommand{\MeasSymbol}   { {X} }
\newcommand{\meassymbol}   { {x} }

\newcommand{\Past}	{ \MeasSymbol_{:0} }

\newcommand{\Future}	{ \MeasSymbol_{0:} }

\newcommand{\CausalState}	{ \mathcal{S} }

\newcommand{\causalstate}	{ \sigma }

\newcommand{\AlternateState}	{ \mathcal{R} }

\newcommand{\Prob}      {\Pr} 

\newcommand{\Cmu}		{C_\mu}
\newcommand{\hmu}		{h_\mu}
\newcommand{\EE}		{{\bf E}}

\newcommand{\PC}		{\chi}

\newcommand{\forward}{+}
\newcommand{\reverse}{-}
\newcommand{\forwardreverse}{\pm} 

\newcommand{\FutureCausalState}	{ {\CausalState}^{\forward} }

\newcommand{\PastCausalState}	{ {\CausalState}^{\reverse} }

\newcommand{\lastindex}[2]{
  \edef\tempa{0}
  \edef\tempb{#2}
  \ifx\tempa\tempb
    \edef\tempc{#1}
  \else
    \edef\tempa{0}
    \edef\tempb{#1}
    \ifx\tempa\tempb
      \edef\tempc{#2}
    \else
      \edef\tempc{#1+#2}
    \fi
  \fi
  \tempc
}
\newcommand{\BE}[2][0]{
  \ensuremath{H[\MeasSymbol_{{#1}:{#2}}]}
}
\newcommand{\BSE}[2][0]{
  \ensuremath{H[\MeasSymbol_{{#1}:{#2}} \CausalState_{#2}]}%
}
\newcommand{\SBE}[2][0]{
  \ensuremath{H[\CausalState_{#1} \MeasSymbol_{{#1}:{#2}}]}
}

\newcommand{\AltSBE}[2][0]{
  \ensuremath{H[\AlternateState_{#1} \MeasSymbol_{{#1}:{#2}}]}
}
\newcommand{\AltBSE}[2][0]{
  \ensuremath{H[\MeasSymbol_{{#1}:{#2}} \AlternateState_{#2}]}
}

\newcommand{\COrder}{k}

\newcommand{\MOrder}{R}

\newcommand{\CSjoint}[1][,]{
   \edef\tempa{:}
   \edef\tempb{#1}
   \ifx\tempa\tempb
      \ensuremath{\FutureCausalState\!#1\PastCausalState}
   \else
      \ensuremath{\FutureCausalState#1\PastCausalState}
   \fi
}

\newif\ifpm 
\edef\tempa{\forwardreverse}
\edef\tempb{\pm}
\ifx\tempa\tempb
   \pmfalse
\else
   \pmtrue  
\fi

\newcommand{\MS} {\MeasSymbol}

\begin{document}

\title{How Hidden are Hidden Processes?\\
A Primer on Crypticity and Entropy Convergence}

\author{John R. Mahoney}
\email{jmahoney3@ucmerced.edu}
\affiliation{Physics Department,
University of California at Merced,\\
5200 North Lake Road Merced, CA 95343}

\author{Christopher J. Ellison}
\email{cellison@cse.ucdavis.edu}
\affiliation{Complexity Sciences Center\\
Physics Department, University of California at Davis,\\
One Shields Avenue, Davis, CA 95616}

\author{Ryan G. James}
\email{rgjames@ucdavis.edu}
\affiliation{Complexity Sciences Center\\
Physics Department, University of California at Davis,\\
One Shields Avenue, Davis, CA 95616}

\author{James P. Crutchfield}
\email{chaos@cse.ucdavis.edu}
\affiliation{Complexity Sciences Center\\
Physics Department, University of California at Davis,\\
One Shields Avenue, Davis, CA 95616}
\affiliation{Santa Fe Institute,
1399 Hyde Park Road, Santa Fe, NM 87501}

\date{\today}
\bibliographystyle{unsrt}

\begin{abstract}
We investigate a stationary process's crypticity---a measure of the difference
between its hidden state information and its observed information---using
the causal states of computational mechanics. Here, we
motivate crypticity and cryptic order as physically meaningful
quantities that monitor how hidden a hidden process is.
This is done by recasting previous results on the convergence of block entropy and block-state
entropy in a geometric setting, one that is more intuitive and that leads
to a number of new results. For example, we connect crypticity to how an
observer synchronizes to a process. We show that the block-causal-state entropy
is a convex function of block length. We give a complete analysis of spin chains. We present a classification scheme that surveys stationary processes in
terms of their possible cryptic and Markov orders. We illustrate related entropy convergence behaviors using a new form of foliated information diagram.
Finally, along the way, we provide a variety of interpretations of crypticity
and cryptic order to establish their naturalness and pervasiveness. Hopefully,
these will inspire new applications in spatially extended and network
dynamical systems.

\vspace{0.1in}
\end{abstract}

\preprint{Santa Fe Institute Working Paper 11-08-XXX}
\preprint{arxiv.org:1108.XXXX [physics.gen-ph]}

\maketitle
\

\tableofcontents
\setstretch{1.1}

\newpage
{\bf
A black box is a metaphor for ignorance: One cannot see inside, but the
presumption is that something, unknown in whole or in part, is there to be
discovered. Moreover, the conceit is that the impoverished outputs from the
box do contain something partially informative.
Physically, ignorance comes in the act of measurement---measurements that
are generically incomplete, inaccurate, and infrequent.
Since measurements dictate that one can have only a partial view, it goes 
without saying that these
distortions make discovery both difficult and one of the key challenges
to scientific methodology.
Measurement necessarily leads to our viewing the world as being hidden
from us. Of course, the world is not completely hidden. If it
were, then there would be neither gain nor motivation to probing
measurements to build models.
Scientific theory building and its experimental verification operate, then,
in the framework of hidden processes---processes from which we have
observations from which, in turn, we attempt to understand the hidden mechanisms.
At least philosophically, this setting is not even remotely new.
The circumstance is that addressed by Plato's metaphor of our
knowledge of the world deriving from the data of shadows on a cave wall.

Fortunately, we are far beyond metaphors these days.
Hidden processes pose a quantitative question: How hidden are they? Here,
we show how to quantitatively measure just this: How much
internal information is hidden by
measuring a process? Of course, this assumes, as in the black box metaphor,
that there is something to be discovered. The tool we use to
ground the intentional stance of discovering the internal
mechanisms---to say \emph{what} is hidden---is computational
mechanics. Computational mechanics is a theory of what patterns
are and how to measure a hidden process's degree of structure and organization.
Computational mechanics has a long history, though, going back to
the original challenges of nonlinear modeling posed in the 1970s
that led to the concept of reconstructing ``geometry from a time
series''.
The explorations here can be seen in this light, with one
important difference: Computational mechanics shows that
measurements of a hidden process tell how the process's internal
organization should be represented. Building on this, we
develop a quantitative theory of how hidden processes are.
}

\section{Introduction}

Many scientific domains face the confounding problems of defining and
measuring information processing in dynamical systems. These range from
technology to fundamental science and, even, epistemology of science
\cite{Crut10b}:
\begin{enumerate}
\item \emph{The 2020 Digital Roadblock}: The end of Moore's scaling laws for
	microelectronics \cite{Moor65a,Moor75a,Moor95a}.
\item \emph{The Central Dogma of Neurobiology}: How are the intricate physical,
	biochemical, and biological components structured and coordinated to
	support natural, intrinsic neural computing?
\item \emph{Physical Intelligence}: Does intelligence require biology, though?
	Or can there be alternative nonbiological substrates which support system
	behaviors that are to some degree ``smart''.
\item \emph{Structure versus Function}: Intelligence aside, how do we define
	and detect spontaneous organization, in the first place?
	How do these emergent patterns take on and support functionality?
\end{enumerate}
Many have worked to quantify various aspects of information
dynamics; cf. Ref. \cite{Atma91a}. One
often finds references to information storage, transfer, and processing.
Sophisticated measures are devised to characterize these quantities in
multidimensional settings, including networks and adaptive systems.

Here, we investigate foundational questions that bear on all these domains,
using methods
with very few modeling and representation requirements attached
that, nonetheless, allow a good deal of progress. In quantifying
information processing in stochastic dynamical systems, two measures have
repeatedly appeared and been successfully applied: the past-future mutual
information of observations (excess entropy) 
$\EE$~\cite[and references therein]{Crut01a} and
the internal stored information (statistical complexity) $\Cmu$
\cite{Crut88a}. Curiously, the difference between these measures---the
crypticity $\PC$ \cite{Crut08a}---has only recently received attention. To 
our knowledge,
the first attempt to understand $\PC$ directly was in
Ref.~\cite{Maho09a}.
The following provides additional perspective and clarity to the
results contained there and in the related works of
Refs.~\cite{Jame11a,Crut08a,Crut08b}. In particular, we add to the
body of knowledge surrounding crypticity and cryptic order,
develop a further classification of the space of processes, and
introduce several alternative ways to visualize these concepts.
An appendix demonstrates that crypticity captures a notable and unique
property, when compared to alternative information measures.
The goal is to provide a more intuitive and geometric toolbox for
posing and answering the increasing range and increasingly more
complex research challenges surrounding information processing
in nature and technology.

\section{Definitions}

We denote contiguous groups of random variables
$\MS_i$ using $\MS_{n:m+1} = \MS_n \ldots \MS_m$. A semi-infinite group is 
denoted either $\MS_{n:} = \MS_n \MS_{n+1} \ldots$ or
$\MS_{:n} = \ldots \MS_{n-2} \MS_{n-1}$.
We refer to these as the \emph{future} and the \emph{past},
respectively. Consistent with this, the bi-infinite chain of random
variables is denoted $\MS_{:}$. A \emph{process} is
specified by the distribution $\Prob(\MS_{:})$. Throughout
the following, we assume we are given a stationary process.

Please refer to Refs.~\cite{Maho09a, Crut10a} for supplementary definitions
of presentations, causal states, \eMs, unifilarity, co-unifilarity, Shannon
block information, information diagrams, and the like. The following assumes
familiarity with these concepts and the results and techniques there.
However, our development calls for a few reminders.

There are two notions of memory central to characterizing stochastic
processes. These are the excess entropy $\EE$ (sometimes called the predictive
information) and the statistical complexity $\Cmu$. The excess
entropy is a measure of correlation between the past and future:
the degree to which one can remove uncertainty in the future given
knowledge of the past. (This is illustrated as the green
information atom at the intersection the past and future in the
information diagram of Fig.~\ref{fig:CrypticityDef}.) The
statistical complexity is a quantity that arises in the context of
modeling rather than prediction. Specifically, $\Cmu$ is the
amount of information required for an observer to synchronize a
stochastic process. In the setting of finite-state hidden Markov
models, it is the information stored in the process's causal
states.

Then, we have the crypticity:

\begin{Def}
A process's \emph{crypticity} $\PC$ is defined as:
\begin{align*}
\PC = H[\CausalState_0 | \MS_{0:}] ~,
\end{align*}
where $\CausalState_t$ is a process's causal state at time $t$.
\label{def:Crypticity}
\end{Def}

Clearly, the definition relies on having a process's \eM\ presentation;
the states used are causal states. Other presentations, whose
alternative states we denote $\AlternateState$, suggest an
analogous, but more general
definition of crypticity; cf. Ref. \cite{Crut10a}.

\begin{figure}[ht]
\centering
\includegraphics{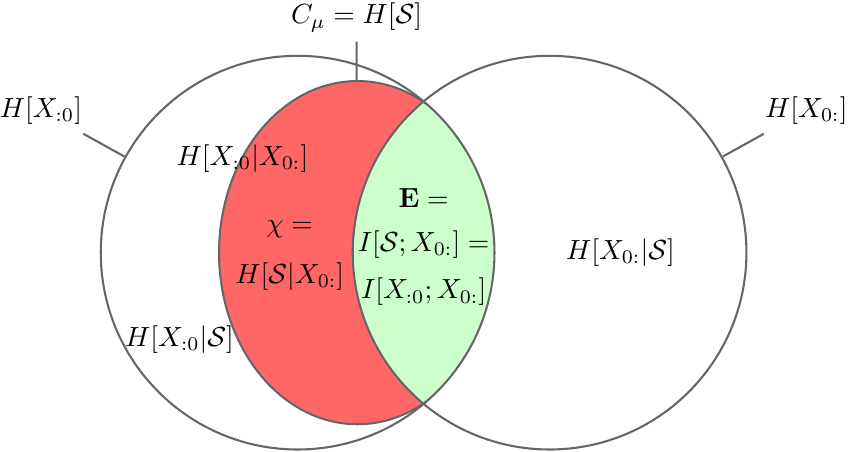}
\caption{Crypticity $\PC$ is represented by the red (dark) crescent shape in
  this \eM\ I-diagram. The excess entropy $\EE$, by the (green) overlap of
  the past information $H[\Past]$ and future information $H[\Future]$.
  The statistical complexity $\Cmu$ is the information in the internal causal
  states $\CausalState$ and comprises both $\PC$ and $\EE$.
  For a review of information measures and diagrams refer to the citations
  given in the text or quickly read the first portions of
  Sec. \ref{sec:IDiagrams}.
  }
\label{fig:CrypticityDef}
\end{figure}

To give us something to temporarily hang our hat on, it turns out that the
crypticity is simply how much stored information is hidden from observations.
That is, it is the difference between the internal stored information
($\Cmu$) and the apparent past-future mutual information ($\EE$).
This is directly illustrated in Fig.~\ref{fig:CrypticityDef}.

We are also interested in the range required to ``learn'' the
crypticity. This is the \emph{cryptic order}.

\begin{Def}
A process's \emph{cryptic order} $k$ is defined as:
\begin{align*}
k = \min \{ L \in \mathbb{Z}^+ :
  H[\CausalState_L | \MS_{0:}] = 0 \} ~.
\end{align*}
\label{def:CrypticOrder}
\end{Def}

These definitions do not easily admit an intuitive interpretation.
Their connection to hidden stored information is not immediately
clear, for example. They mask the importance and centrality of the crypticity 
property. Given this, we devote some effort in the following to motivate them 
and to give several supplementary interpretations.
As a start, Fig. \ref{fig:CrypticityDef} gives a graphical definition of
crypticity using the \eM\ information diagram of Ref. \cite{Crut08a}.
It is the red crescent highlighted there, which is the state information
$\Cmu = H[\CausalState]$ minus that information derivable from the
future $H[\Past] \equiv H[\MS_{0:}]$.
This begins to explain crypticity as a measure of a process's
hidden-ness. We'll return to this, but first let's consider
several other alternatives.

\section{Crypticity: From state paths to synchronization}
\label{sec:Paths}

Crypticity and, in particular, cryptic order have straightforward 
interpretations when one considers the internal state-paths taken as an observer
synchronizes to a process \cite{Crut01b}. In this, cryptic order is seen to be
analogous to, and potentially simpler than, a process's Markov order. While
both the Markov and cryptic orders derive from a notion of synchronization,
the cryptic order depends on a subset of the paths realized during
synchronization. We illustrate this via an example: The $(R,k)$-Golden Mean
Process---a generalization of the Golden Mean Process with tunable
Markov order $R$ and tunable cryptic order $k$. In particular, we
examine the $(3,2)$-Golden Mean Process shown in Fig.~\ref{fig:RkGoldenMean}.

\begin{figure}[ht]
\includegraphics{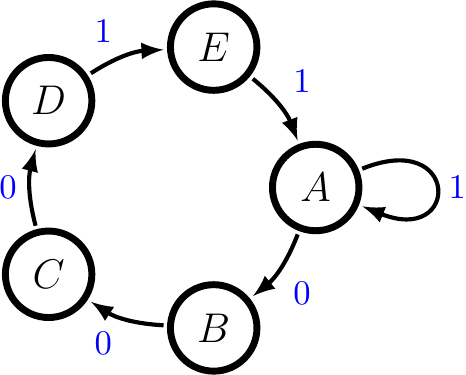}
\caption{The $(3,2)$-Golden Mean \eM: Markov order $3$ and cryptic order $2$.
  }
\label{fig:RkGoldenMean}
\end{figure}

It is straightforward to verify that the only words of length 3
generated by this process are $\{000,001,011,100,110,111\}$. Since
the process is Markov order $3$ (by construction) we know that
each of these words is a synchronizing word \footnote{A
\emph{synchronizing word} is symbol sequence that induces one and
only one causal state. This is in contrast with a \emph{minimal}
synchronizing word, of which no proper prefix is a synchronizing
word.}.
Some words lead to synchronization in fewer than three steps,
though. For instance, $011$ yields synchronization to state $E$ after 
just the first two symbols $01$.

In Fig.~\ref{fig:paths}, we display the internal-state paths taken by each
possible initial state under evolution governed by the six
synchronizing words. Let's take a moment to describe these
illustrations carefully. Before reading any word, there is maximum
uncertainty in the internal state. We represent this using a
circle for each of the five causal states of the \eM. Each of
these states is led to a next state by following the first symbol
seen \footnote{The fact that any state is led by any symbol to at
most one next state is the property known as unifilarity---a
direct consequence of the states being causal states.}.
For word $001$, the first symbol is $0$, and $A$, for instance,
is led to $B$. Notice that $E$ is not led to any state. This is
because $E$ has no outgoing transition with symbol $0$. The path
from $E$, therefore, ends and is not considered further. The
termination of paths is one of the important features of
synchronization to note.

Looking at the synchronizing word $100$, we see that the transition on the
first symbol $1$ takes both states $A$ and $E$ to the same state $A$. Since
we use unifilar presentations (\eMs), this merging can never be undone.
Path merging is yet another important feature.

Both the termination and merging of paths are relevant to synchronization, 
but have different roles in the determination of the Markov and cryptic orders.

Although we already know the Markov order of this process, we can
read it from Fig.~\ref{fig:paths} by looking at the lengths for
each word where only one path remains. These lengths
$\{3,3,2,2,2,2\}$ are marked with orange diamonds. The maximum value
of this length is the Markov order ($3$, in this example).

\begin{figure}[ht]
\includegraphics{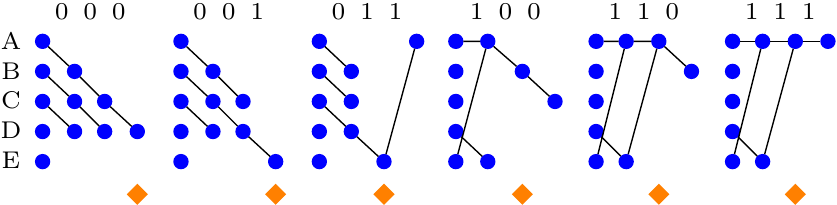}
\caption{Synchronization paths for $(3,2)$-Golden Mean \eM: Each
  synchronizing word induces a set of state-paths; some of which
  terminate, some of which merge.
  }
\label{fig:paths}
\end{figure}

In the next illustration, Fig.~\ref{fig:paths_cryptic}, we keep
only those paths that do not terminate early. In this way, we
remove paths that generally are quite long, but that terminate
before having the chance to merge with the final synchronizing
paths. We similarly mark, with green triangles,
the lengths where these reduced paths have ultimately merged. Note
that restricting paths can only preserve or decrease each length. Finally, 
in analogy to the Markov order, the maximum of these lengths \{0,0,0,1,2,2\} 
is the cryptic order (2 in this example).

\begin{figure}[ht]
\includegraphics{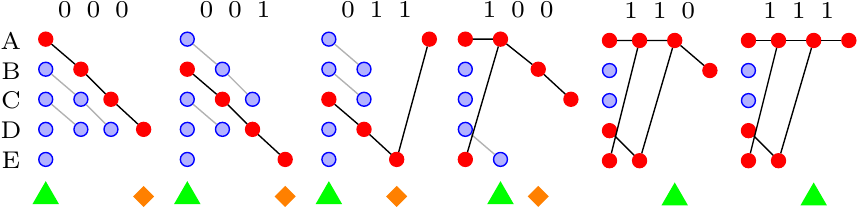}
\caption{The paths that are not terminated before the Markov order are
  highlighted in red. These are the paths relevant for the cryptic order.
  For each word the contribution to Markov order is still indicated by an
  orange diamond, whereas the contribution toward cryptic order is indicated
  by a green triangle.
  }
\label{fig:paths_cryptic}
\end{figure}

This demonstrates how crypticity relates to paths and path merging. It is a
small step then to ask for a direct connection to co-unifilarity
\cite{Maho09a}: $H[\CausalState_0 | \MS_0 \CausalState_1] = 0$. In fact, there
are three primary equivalent statements about a process: (i) its \eM\ being
co-unifilar, (ii) its $\PC = 0$, and (iii) its cryptic order $k = 0$. (Appendix~\ref{app:counif}
presents a proof of this equivalence in terms of entropy growth
functions and includes the connection to cryptic order as well.)

This exposes the elementary nature of the cryptic order as a property of
synchronizing paths. Appendix \ref{sec:RetroPath} goes further to show that
state-paths traced similarly, but in the reverse time direction, are the same
as those singled-out in the forward direction, as just done. The remainder of
this section offers different perspectives on crypticity, some of which are
less strict, but provide intuition and suggest its broad applicability.

\subsection{Global versus Local}
Imagine a synchronization task involving a group of agents. The agents begin 
in different locations (states) and move to next locations based on the 
synchronization input they receive from a common controller. The goal is to 
provide a uniform input that causes (a subset of) the agents to arrive at the 
same location. This is reminiscent of a road-coloring problem. In many 
road-coloring contexts, only uniform-degree graph structures are investigated, 
largely due to theoretical tractability. However, real-world graphs are rarely 
uniform degree. This means that some agents may receive instructions that they 
cannot carry out. These agents quit, and their paths are terminated. Assuming 
that the instructions are synchronizing for some subset of the agents 
(the instruction is a synchronizing word), the synchronization task will end 
with this subset of agents at the desired destination.

There are two ways in which we may view this process. One is global, and 
corresponds to the Markov order, while the other is local and corresponds to 
the cryptic order.

If we monitor the entire collection of agents from a bird's eye view evolving 
under the synchronization input, we observe paths terminating and merging. 
Our global notion of synchronization is the point at which each path is either 
terminated or merged with every other valid path. This is clearly coincident 
with the description of Markov order previous described.

Alternatively, we monitor the collective by querying the agents
after the task is complete. The unsuccessful agents, whose paths were 
terminated, never arriving at the destination, cannot be queried. From this 
viewpoint, synchronization takes place \emph{relative} to the group of agents 
that were not terminated. As locally interacting entities, they know the latest 
time at which an agent merged with their group---the group which ultimately 
synchronized. Even after this event, there may be other agents still operating 
that will inevitably be terminated at some later time. This means that from the 
local (agent) perspective, synchronization may happen earlier than from the 
global (controller) perspective.

We claim, based on this setting, that the cryptic order has a straightforward 
and physically relevant basis in the context of synchronization. Upcoming 
discussions, some more technical, will emphasize this point
further, as well as demonstrate new results.

\subsection{Mazes and Stacks}

The Markov versus cryptic order distinction is relevant to any
maze-solving algorithm \footnote{The maze example is not a stationary process,
so there are some important differences. For instance, there can be words
that terminate \emph{after} the length of the longest maze solution.}.
Imagining the solution of a maze as a sequence of moves---left,
right, or straight---we may write
down a list of potential solutions (which must contain all actual
solutions) by listing all $3^{N^2}$ sequences \footnote{Assume an
$N$ by $N$ maze. A nonintersecting solution cannot contain more
instructions than there are locations within the maze.}.
A brute-force algorithm tries all of these paths. Since we are
interested in worst-case scenarios, many of the details (e.g., depth-
versus breadth-first search) are not relevant. What \emph{is}
relevant is the object that the algorithm must maintain in memory
or that it ultimately returns to the user.

An algorithm might try out each potential solution, feeding in each move 
sequentially and testing for either maze completion or termination (walking 
into a wall or a previously visited location) at each step. The end of each 
solution is marked with a length. When all solutions have been tried, this 
set of solutions and lengths is returned. While this is not a stationary 
stochastic process, we may think of the longest of these lengths as being 
similar to the Markov order. The speed and memory use of this algorithm are 
obviously improved by using a tree structure, but this does not affect the 
result we are interested in.

If we were only interested in paths which end in maze completion, an even more
memory-conscious algorithm would realize that dead-ends in the tree could be 
removed. One accomplishes this with a stack memory for the active-path tree
branch. Reaching a nonsolving termination, the algorithm pops the
end states until returning to the most recent unexplored option.
This process continues recursively until the tree has been filled
out. The relevant lengths are now the lengths of the
maze-completing paths (all root-to-leaf paths), the longest of
which is an analog of the cryptic order.


\subsection{Transient versus Relaxed}

Rather than using the global versus local distinction, we can think in terms
of a dynamical view of synchronization. We might imagine a collection of ants 
attempting to create paths from a resource-rich region to their nest; or a 
watershed in the process of forming the transport network from collection 
regions to the main body of water. Until these networks develop, it is not 
clear which will become the important paths.

A log not worth climbing over causes ants to make the effort less often,
thereby dropping less pheromone, leading fewer ants to attempt this path,
until finally it is empty. Similarly, slow water deposits more sediment and
fills underused channels. As these networks evolve from an initial transitory
state to relaxed state, the types of paths within the network and their
synchronization properties change. In particular, while the
early-time synchronization depends on the terminating paths, the
later-time synchronization will not. In this dynamical picture we
see that a property akin to cryptic order emerges as the system evolves.

\subsection{Naive versus Informed}

It is only a small step from this dynamical picture to view these
self-reinforcing systems as evolving from naive to informed states. Over
time, a system ``realizes'' which paths are undesirable and quits them.
Consider an individual learning to navigate a new city. She will
experience a similar network evolution, where the pruning of
dead-end paths is an intentional act. This navigation structure
also will tend to reflect the cryptic order.

\subsection{Statistical Complexity versus Crypticity}

In addition to describing the Markov and cryptic orders via a
dynamical picture of synchronization, we can explore the same
phenomenon with the associated entropies, a more statistical
perspective.

Beginning with the global view, the distribution over the set of
all starting points is the state entropy $H[\CausalState_0]$,
commonly called the statistical complexity $\Cmu$. By considering
the initial state distribution conditioned on the removal of the
terminating paths, we are left with only a portion of this
entropy, and this is the crypticity $\PC$ \footnote{To be more precise, it
is not so much that the statistical complexity is derived from
considering all paths as it is derived from considering \emph{no}
paths.}. As discussed, we might consider this removal a result of memory, 
relaxation, or prescience. 

\section{Crypticity through Information Theory}

The discussion above in terms of paths is relatively intuitive.
The original conception, however, was not in terms of paths, but
rather in terms of information-theoretic quantities. Information
identities based on \eMs\ are beginning to provide a growing set
of interpretations; some more subtle, some more direct than others.
The following will show that crypticity and cryptic order have diverse
implications and also that even elementary information-theoretic
quantities form a rich toolset.

\subsection{Crypticity}

The \eM\ causal presentation pairs up pasts with futures in a way appropriate
for prediction. Since pasts can be different but predictively equivalent, this
pairing operates on sets of pasts that, in turn, are equivalent to the causal
states themselves. Furthermore, a single past can be followed by a set of
futures. This is natural since the processes are stochastic. So, any past or 
predictively equivalent group of pasts is linked to a distribution of futures.
Finally, these future distributions often overlap. As we will now
show, crypticity is a measure of this overlap.

Historically, it has taken some time to sort out the similarities and
differences between various measures of memory. Eventually, two emerged
naturally as key concepts: $\Cmu$, the statistical complexity or
information processing ``size'' of the internal mechanism;
and $\EE$, the excess entropy, or the apparent (to an observer)
amount of past-future mutual information. It has been recognized for some time
\cite{Crut97a,Shal98a} that $\Cmu$ is an upper bound on
$\EE$. The strictness of this inequality and the nature of
the relationship between the two, however, was not significantly
explored until Ref.~\cite{Maho09a}. The first simple statement
\cite{Crut08a} about crypticity in terms of information-theoretic
quantities is that it is the quantifiable difference between two
predominant measures of information storage: $\PC = \Cmu - \EE$.

Taking this view a bit further, since $\EE$ is the amount of
uncertainty in the future that one can reduce through study of the
past, and $\Cmu$ is the amount of information necessary to do
optimal prediction (using a minimal predictor), their difference is
the amount of \emph{modeling overhead}. One may object that a minimal
optimal predictor should not require more information than will be
made use of. In fact it is known that many processes with large
$\PC$ have nonunifilar representations that are much smaller \cite{Maho09a}.
What is not obvious is that this is simply a re-representation of the causal
states as mixtures of the new states \cite{Crut08b,Maho09b}.
In other words, \emph{the overhead is inescapable.} This suggests a useful language with
which to discuss stochastic processes---not only do we identify a
process with an $\epsilon$-\emph{machine}, but we analyze the efficiency
of these machines in terms of required resources.

For the following, we briefly invoke the use of the reverse \eM, the causal
representation of a process when scanned in reverse, to extend our
view of crypticity. (For details on reverse causal states,
see Refs. \cite{Crut08a,Crut08b}.) Recall that forward causal
states are built for prediction and, similarly, reverse causal
states are built for retrodiction. We say they are ``built'' for
these purposes in the sense that they are minimal and optimal, two desirable design goals. 

Given this, it is somewhat surprising to see that forward causal states are
better at retrodiction than reverse causal states. The information diagram
in Fig. \ref{fig:BetterRetro} illustrates this. We will now show that the
degree to which this is true is precisely the forward process's crypticity.
Here, we write this difference in retrodictive uncertainty as follows:
\begin{align*}
H[\MS_{-L:0} | \PastCausalState_0] - H[\MS_{-L:0} | \FutureCausalState_0] &\ge 0 
	~.
\end{align*}
Then, this difference converges to $\PC$:
\begin{align*}
\PC = \lim_{L \to \infty} 
	\bigl( 
	H[\MS_{-L:0} | \PastCausalState_0] - H[\MS_{-L:0} | \FutureCausalState_0]
	\bigr) ~.
\end{align*}

We might wonder why the reverse causal states were not built to be
better at their job. This is explained by the fact that the
information input to the above constructs is not equivalent. The
forward causal states are built from the past, while the reverse
causal states are built from the future. It is no surprise, then, that the 
forward states can offer information about the pasts from which they were 
built. It is more interesting to consider why they do not maintain all of this 
information. This is because the forward states were designed for predicting
a stochastic process, a goal for which maintaining information about the past offers 
diminishing returns.

\begin{figure}
\includegraphics[width=\linewidth]{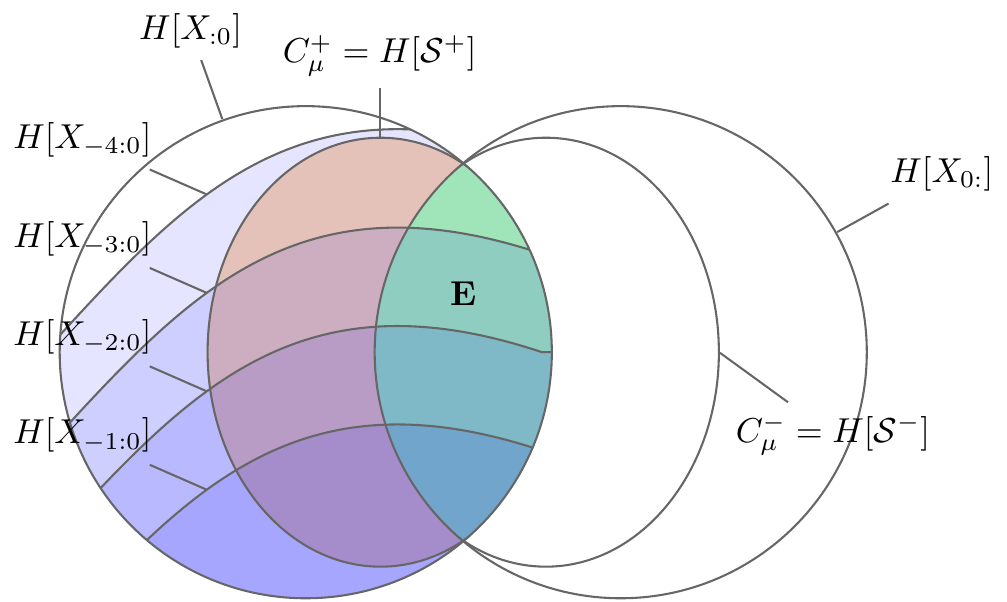}
\caption{Crypticity as the degree to which forward causal states are
  better retrodictors than reverse causal states.
  }
\label{fig:BetterRetro}
\end{figure}

Rather than comparing the function of two objects (forward and
reverse causal states), we can compare two functions of the same
object. In this light, the crypticity is the degree to which
forward causal states are better at retrodiction than they are at
prediction. More precisely, we have:
\begin{align*}
H[\MS_{0:L} | \CausalState_0] & - H[\MS_{0:L} | \CausalState_L]\\
&= H[\CausalState_0 \MS_{0:L}] - H[\MS_{0:L} \CausalState_L] \\  
&= H[\CausalState_0 | \MS_{0:L} \CausalState_L] - H[\CausalState_L | \CausalState_0 \MS_{0:L}] \\
&= H[\CausalState_0 | \MS_{0:L} \CausalState_L] \\
&\ge 0 ~.
\end{align*}
The first step follows from stationarity, the second appeals to an
informational identity, and the next to unifilarity of the \eM.
Similarly, this difference converges to $\PC$:
\begin{align*}
\lim_{L \to \infty}{H[\CausalState_0 | \MS_{0:L} \CausalState_L]} 
&= \lim_{L \to \infty}{H[\CausalState_0 | \MS_{0:L}]} =\PC
\end{align*}
Thus, crypticity is the amount of information that, although necessary for 
current prediction, must be erased at some future time.

\subsection{Cryptic Order}

Many of these statements about uncertainty can be rephrased in terms of 
length scales. The length scale associated with the crypticity is the cryptic 
order: the distance we must look into the past to discover the
modeling overhead. Following our discussion of forward 
and reverse states, we can interpret
cryptic order as the length at which the difference converges to $\PC$:
\begin{align*}
 k = \min\{ L : H[\MS_{-L:0} | \PastCausalState_0] 
 				- H[\MS_{-L:0} | \FutureCausalState_0] = \PC \} ~.
\end{align*}
Stated differently, it is the length at which all advantage of a forward state 
over a reverse state as a retrodictor is lost. In other words:
\begin{align*}
k = \min \{ L :  H[\MS_0 | \MS_{1:L+1} \FutureCausalState_{L+1} ] =
  H[\MS_0 | \MS_{1:L+1} \PastCausalState_{L+1} ] \} ~.
\end{align*}

Equivalently cryptic order is the length at which a forward state's uncertainty
in prediction and retrodiction equalize. More colloquially, it is the range
beyond which a forward state is equally good at prediction and
retrodiction, or:
\begin{align*}
k = \min \{  L : H[\MS_L | \CausalState_0 \MS_{0:L}] = H[\MS_0 |
  \MS_{1:L+1} \CausalState_{L+1} ] \} ~.
\end{align*}

As Sec.~\ref{sec:Paths} suggested, the cryptic order $k$ is closely
analogous to the Markov order $R$. Here, we state the parallel formally:
\begin{align*}
    R &= \min \{ L : H[\CausalState_L | \MS_{0:L}] = 0 \}\\
    k &= \min \{ L : H[\CausalState_L | \MS_{0:L}, \MS_{L:} ] = 0\} ~.
\end{align*}
Appendix \ref{sec:WhyCrypticity} argues for the uniqueness of this parallel.

Cryptic order is the largest noninferable state sequence. Given an infinite
string of measurements $\ldots \meassymbol_{-2} \meassymbol_{-1} \meassymbol_0$,
one eventually synchronizes to a particular causal state \cite{Trav10b}, for
any finite-state \eM. The same symbol sequence can then be used to retrodict 
states beginning at the point of synchronization. All but the earliest
$\COrder$ states can be definitively retrodicted regardless of which observed
sequence (and resulting predictive state) occurs.

\section{Crypticity and Entropy Convergence}

It has become increasingly clear that entropy functions are useful characterizations of processes. Since a process is a bi-infinite collection of
random variables \footnote{One might choose to consider the process to be a
bi-infinite collection of symbols or consider including causal states as well.
We could similarly consider reverse causal states.}, it typically is not useful
to calculate the entropy of the entire collection. The alternative strategy is
to analyze the entropy of increasingly large finite portions. The scaling,
then, captures the system's bulk properties in the large-size (thermodynamic)
limit, as well as how those properties emerge from the individual components.

These functions capture much of the behavior that we are interested in here.
The block entropy $\BE{L}$ was used to great effect in Ref.~\cite{Crut01a}
to understand the way perceived randomness may be reformulated as structure,
when longer correlations are considered. More recently, Ref.~\cite{Crut10a} used
extended functions---the block-state entropy $\AltBSE{L}$ and the state-block
entropy $\AltSBE{L}$---to explore the relationship between alternate
presentations of a process and the information theoretic measures of
memory in a presentation.

We will borrow these two new entropy functions and turn them back on the
canonical set of presentations, \eMs, to expose the workings of crypticity.
The result is a graphical approach that offers a more intuitive understanding
of the results originally developed in Ref.~\cite{Crut08b}. Using this, we
sharpen several theorems, discover new bounds, and pose additional challenges.

\subsection{Block Entropy}

The \emph{block entropy} $H[\MS_{0:L}]$ is the joint Shannon entropy of
finite sequences. As it is treated rather thoroughly in Ref.~\cite{Crut01a},
we simply recall several of its features.

First, recall that $\MS_{0:0}$ represents the random variable for a null
observation and, since there is just one way to do this, $H[\MS_{0:0}] = 0$.
As $L$ increases, the block entropy curve is a nondecreasing,
concave function that limits to the linear asymptote $\EE + \hmu L$, where
$\EE$ is the excess entropy and $\hmu$ is the process entropy rate.

Given a block entropy curve, Markov processes are easily identified since
the curve reaches it linear asymptote at finite block length. That is,
the Markov order $R$ satisfies:
\begin{align*}
	R \equiv \min \, \{ L : H[\MS_{0:L}] = \EE + \hmu L \} ~.
\end{align*}

Before reaching the Markov order, one has not discovered all process
statistics and, so, new symbols appear more surprising than they otherwise 
would. Mathematically, this is formulated through a lower bound:
\begin{align*}
	H[\MS_L | \MS_{0:L}] \geq \hmu ~,
\end{align*}
for all $L$.  Since the block entropy curve for Markovian processes reaches its
asymptote at $L=R$ and since the linear asymptote has slope equal to the entropy
rate, we know that Markov processes attain the lower bound whenever $L > R$:
$H[\MS_L | \MS_{0:L}] = \hmu$.

Finally, since the block entropy is concave and nondecreasing, it is
bounded above by its linear asymptote.  This naturally leads to a concave,
nondecreasing lower bound estimate for the excess entropy:
\begin{align*}
	\EE(L) \equiv H[\MS_{0:L}] - \hmu L~.
\end{align*}
Thus, $\EE(L) \leq \EE(L+1) \leq \EE$ and $\lim_{L\to\infty} \EE(L) = \EE$.

\subsection{State-block entropy}

The state-block entropy $\AltSBE{L}$ is the joint uncertainty one has in a
presentation's internal state $\AlternateState$ and the block of symbols
immediately following. Its behavior is generally nontrivial, but when
restricted to \eMs, its behavior is simple~\cite{Crut10a}. In that case,
it refers to the process's unknown causal state $\CausalState_0$
and is denoted $\SBE{L}$.

Its simplicity is a direct consequence of the causal states' efficient encoding
of the past.  To see this, note that differences in the state-block
entropy curve, the rate at which it grows with block length, are constant:
\begin{align*}
	\SBE{L+1} - \SBE{L}
	&= H[\MS_L | \CausalState_0 \MS_{0:L}] \\
	&= H[\MS_L | \CausalState_{0:L} \MS_{0:L}]\\
	&= H[\MS_L | \CausalState_L]\\
	&= \hmu ~.
\end{align*}
Here, we used the unifilarity property of \eMs:
$H[\CausalState_{L+1} | \CausalState_L, \MS_L] = 0$.
So, given the causal state $\CausalState_0$, the block $\MS_{0:L}$ of
symbols immediately following it determines each causal state along the way
$\CausalState_{0:L}$. Since causal states are sufficient statistics for
prediction, the future symbol $\MS_L$ depends only on the most recent causal
state $\CausalState_L$ and, finally, the optimality of \eMs\ means
that the next symbol can be predicted at the entropy rate $\hmu$.

In other words, the state-block entropy that employs a process's
\eM\ presentation is a straight line with with slope $\hmu$ and
$y$-intercept $\SBE{0} = H[\CausalState_0] \equiv \Cmu$.
Note that $\SBE{L} \geq \BE{L}$ with equality if and only if
$H[\CausalState_0 | \MS_{0:L}] = 0$. Since conditioning never
increases uncertainty, these two block-entropy curves remain equal from
that point onward.
This necessarily implies that they tend to the same asymptote. So,
if the state-block entropy curve ever equals the block entropy curve, then
the $y$-intercepts of each curve must also be equal: $\Cmu = \EE$.  Stated
differently, the two curves meet if and only if the process has $\PC = 0$.

\subsection{Block-state entropy}
\label{sec:bse}

Finally, we consider the block-state entropy $\AltBSE{L}$, a measure of the
joint uncertainty one has in a block of symbols and the presentation's
subsequent internal state.  Once again, our interest here is with \eMs,
and so we consider $\BSE{L}$. Unlike the state-block entropy, however, the
behavior of this entropy is nontrivial.  We recall a number of its properties
and also establish the equivalence of the cryptic order definitions
given in Refs.~\cite{Maho09a,Crut10a}. Then, we provide a detailed proof of
its convexity, as this does not appear previously.

The block-state entropy begins at $\Cmu$ when $L = 0$. As $L$ increases, the
curve is nondecreasing and tends, from above, to the same linear asymptote
as the block entropy: $\EE + \hmu L$.  Since the state-block entropy is
 $\Cmu + \hmu L$ and since $\Cmu \geq \EE$, we see that the
state-block entropy curve is greater than or equal to the block-state entropy:
$\SBE{L} \geq \BSE{L}$. Equality for $L>0$ occurs if and only if the process
has $\Cmu = \EE$ or, equivalently, $\PC = 0$ and, then, the curves
are equal for all $L$.

Similarly, the block-state entropy is greater than or equal to the block
entropy: $\BSE{L} \geq \BE{L}$. We have equality if and only if
$H[\CausalState_L | \MS_{0:L}] = 0$.  Recall, the smallest such $L$ is the
Markov order $R$. So, the block-state entropy equals the block entropy
only at the Markov order.  Further, once the curves are equal, they remain
equal:
\begin{align*}
	\BSE{L} = \BE{L} \Rightarrow \BSE{L+1} = \BE{L+1}.
\end{align*}
This can shown by individually expanding both $\BSE{L+1}$ and $\BE{L+1}$ to
$\BE{L} + \hmu$. The interpretation is that the two curves become equal only
at the Markov order and only after both curves have reached their linear
asymptotes.

Reference~\cite{Crut10a} defined the cryptic order as the minimum $L$ for which
the block-state entropy reaches its asymptote. This is in contrast
to the definition provided here and also in Ref.~\cite{Maho09a}, which
defines the cryptic order as the minimum $L$ for which
$H[\CausalState_L | \MS_{0:}] = 0$. We now establish the equivalence of
these two definitions.

\begin{The}
\begin{align}
H[\CausalState_L | \MS_{0:}] = 0 \: \iff \: \BSE{L} = \EE + \hmu L ~.
\end{align}
\end{The}

\begin{proof}
\begin{align}
&\hspace{-.15in} H[\CausalState_L | \MS_{0:}] = 0
\label{eq:crypticequiv1}\\
&\iff H[\CausalState_0 | \MS_{0:}] = H[\CausalState_0 | \MS_{0:L}, \CausalState_L]
\label{eq:crypticequiv2}\\
&\iff I[\CausalState_0; \MS_{0:}] = I[\CausalState_0; \MS_{0:L}, \CausalState_L]
\label{eq:crypticequiv3}\\
&\iff \EE = H[\MS_{0:L}, \CausalState_L]
			- H[\MS_{0:L}, \CausalState_L | \CausalState_0]
\label{eq:crypticequiv4}\\
&\iff H[\MS_{0:L}, \CausalState_L] \nonumber \\
	& \quad\qquad = \EE + H[\CausalState_L | \CausalState_0, \MS_{0:L}]
	      + H[\MS_{0:L} | \CausalState_0]
\label{eq:crypticequiv5}\\
&\iff H[\MS_{0:L}, \CausalState_L] = \EE + \hmu L ~.
\label{eq:crypticequiv6}
\end{align}
The step from Eq.~\eqref{eq:crypticequiv1} to Eq.~\eqref{eq:crypticequiv2}
follows from Thm.~1 of Ref.~\cite{Crut10a}.  In moving from
Eq.~\eqref{eq:crypticequiv3} to Eq.~\eqref{eq:crypticequiv4}, we used the
prescience of causal states $\EE = I[\CausalState_0; \MS_{0:}]$~\cite{Shal98a}.
Finally, Eq.~\eqref{eq:crypticequiv5} leads to Eq.~\eqref{eq:crypticequiv6}
using unifilarity of
\eMs\ ($H[\CausalState_L | \CausalState_0, \MS_{0:L}] = 0$) and that they
allow for prediction at the process entropy rate:
$H[\MS_{0:L} | \CausalState_0] = \hmu L$.
\end{proof}

We obtain estimates for the crypticity $\PC$ by considering the difference
between the state-block and block-entropy curves:
\begin{align}
\PC(L) & \equiv \SBE{L} - \BSE{L} \label{eq:PCL}\\
       & = \hmu L - H[\MS_{0:L} | \CausalState_L] \label{eq:PCL2}~.
\end{align}
Ref.~\cite{Maho09a} showed that this approximation limits from below in a
nondecreasing manner to the process crypticity: $\PC(L) \to \PC$ and
$\PC(L) \leq \PC(L+1) \leq \PC$.
This also provides an upper-bound estimate of the excess entropy:
\begin{align*}
	\EE \leq \Cmu - \PC(L) ~.
\end{align*}
Combined with the lower-bound estimate the block entropy provides, one
can be confident in the estimates of excess entropy.

The retrodictive error $H[\MS_{0:L} | \CausalState_L]$ is the difference
of the block-state entropy from the statistical complexity. It is also the
difference of $\PC(L)$ from $\hmu L$.  Furthermore, it follows from
Ref.~\cite{Crut10a} that the asymptotic retrodiction rate~\cite{Crut98d} 
is equal to the process entropy rate:
\begin{align*}
	\lim_{L\to\infty} \frac{H[\MS_{0:L} | \CausalState_L]}{L} = \hmu ~.
\end{align*}
In a sense, this describes short-term retrodiction. As we will see in a moment,
order-$R$ spin-chains are a class of processes that have no retrodiction error
for a full $R$-block. The opposite class, in this sense, consists of processes
with $\PC=0$---that is, the co-unifilar processes. These immediately begin
retrodiction at the optimal rate, which is $\hmu$.

Finally, we establish the convexity of the block-state entropy,
which appears to be new.

\begin{The}
$H[\MS_{0:L} \CausalState_{L}]$ is convex upwards in $L$.
\end{The}

\begin{proof}
Convexity here means:
\begin{align*}
H[\MS_{0:L+1} \CausalState_{L+1}]
  & - H[\MS_{0:L} \CausalState_{L}] \\
  & \ge H[\MS_{0:L} \CausalState_{L}]
    - H[\MS_{0:L-1} \CausalState_{L-1}] ~.
\end{align*}
Stationarity gives us:
\begin{align*}
H[\MS_{-1:L} \CausalState_{L}]
  & - H[\MS_{0:L} \CausalState_{L}] \\
  & \ge H[\MS_{-1:L-1} \CausalState_{L-1}]
  	- H[\MS_{0:L-1} \CausalState_{L-1}] ~.
\end{align*}
Simplifying, we have:
\begin{align*}
H[\MS_{-1} | \MS_{0:L} \CausalState_{L}]
  & \ge H[\MS_{-1} | \MS_{0:L-1} \CausalState_{L-1}] ~.
\end{align*}

We can use the I-diagram of Fig.~\ref{fig:4Variable_IDiagram} to help
understand this last convexity statement. There, it translates into:
\begin{align*}
\alpha + \gamma &\ge \alpha + \beta ~,
\end{align*}
or, since $\alpha \geq 0$:
\begin{align}
\gamma &\ge \beta ~.
\label{eq:ConvexityInIDiagram}
\end{align}

\begin{figure}
\includegraphics{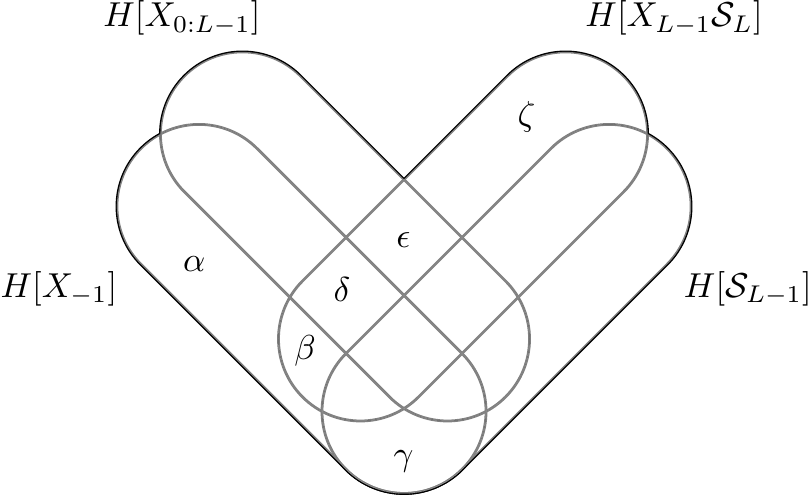}
\caption{Four variable I-diagram for the block-state entropy convexity proof,
  with the needed sigma-algebra atoms appropriately labeled.
  }
\label{fig:4Variable_IDiagram}
\end{figure}

Using the fact that the causal state is an optimal representation of the past,
we have the following expressions that are asymptotically equivalent
to the entropy rate $\hmu$:
\begin{align*}
H[\MS_{L-1} \CausalState_L | \CausalState_{L-1}] &= \beta + \epsilon + \delta + \zeta\\
H[\MS_{L-1} \CausalState_L | \CausalState_{L-1} \MS_{0:L-1}] &= \beta + \zeta\\
H[\MS_{L-1} \CausalState_L | \CausalState_{L-1} \MS_{-1}] &= \epsilon + \zeta \\
H[\MS_{L-1} \CausalState_L | \CausalState_{L-1} \MS_{-1} \MS_{0:L-1}] &= \zeta ~.
\end{align*}
The associations with the sigma-algebra atoms are readily gleaned from the I-diagram.
Note that the various finite-$L$ expressions for the entropy rate rely on the
shielding property of the causal states and also on the \eM's unifilarity.
Taken together in the $L \to \infty$ limit, the four relations yield:
\begin{align*}
\zeta & = \hmu ~\text{and} \\
\beta & = \delta = \epsilon = 0 ~.
\end{align*}
These, in turn, transform the convexity criterion of 
Eq.~(\ref{eq:ConvexityInIDiagram}) into the simple statement that:
\begin{align*}
\gamma \ge 0 ~.
\end{align*}
Since $\gamma =
I[\MS_{-1};\CausalState_{L-1}|\MS_{0:L}
\CausalState_L]$ is a conditional mutual information and,
therefore, positive semidefinite, this establishes that the block-state 
entropy is convex.
\end{proof}

\begin{figure}[ht]
\centering
\includegraphics[width=\linewidth]{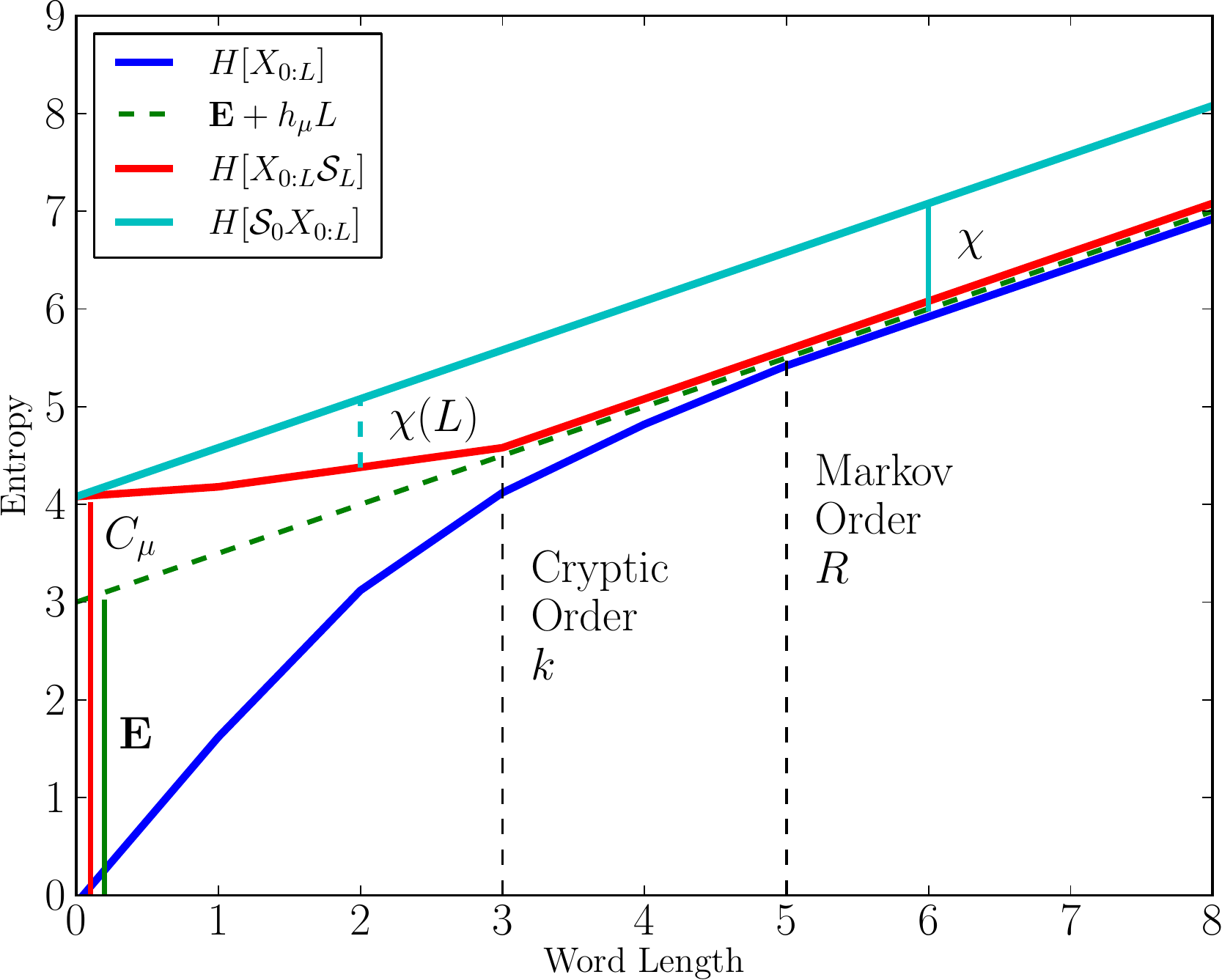}
\caption{The block $H[\MS_{0:L}]$, state-block
  $H[\CausalState_0 \MS_{0:L}]$, and block-state entropy
  $H[\MS_{0:L}\CausalState_0 ]$ curves compared. The sloped dashed
  line is the asymptote $\EE + \hmu L$, to which both the block entropy and
  state-block entropy asymptote. Finite Markov order and finite cryptic order
  are illustrated by the vertical dashed lines that indicate where the
  entropies meet the linear asymptote, respectively.
  The convergence of the crypticity approximation $\PC(L)$ to $\PC$
  is also shown.
  }
\label{fig:EntropyGrowthCurvesCryptic}
\end{figure}

It will help to summarize the point that we have now reached. We used the various
block entropy curves to synthesize much of our information-theoretic viewpoint
of a process into a single representation---that shown in Fig.
\ref{fig:EntropyGrowthCurvesCryptic}. We can amortize the effort to
develop this viewpoint by applying it to a broad class of processes familiar
from statistical mechanics.

\section{Crypticity in Spin Chains}
\label{sec:spinchains}

We first consider a subset of processes drawn from statistical mechanics known
as one-dimensional \emph{spin chains}. (For background, see 
Refs.~\cite{Crut97a,Feld98b}.) They are processes such that 
$H[\MS_{0:\MOrder}] = \Cmu$.
Using this, the simple geometry presented in
Fig.~\ref{fig:EntropyGrowthCurvesCryptic} reveals that:
\begin{align}
  \PC(L) = \begin{cases}
    \hmu L & \quad 0 \leq L \leq \MOrder ~,\\
    \PC & \quad L > \MOrder ~,
  \end{cases}
\end{align}
and this, in turn, implies that $\COrder = \MOrder$. This can also seen through
Eq.~\eqref{eq:PCL2}. Recall, the block-state entropy is nondecreasing and begins
at $\Cmu$.  Since spins chains have $H[\MS_{0:\MOrder}] = \Cmu$, we know that 
the block-state entropy curve for spin chains must remain flat until $L=R$.
Consequently, $H[\MS_{0:L} | \CausalState_L] = 0$ and $\PC(L) = \hmu L$
for $L \leq R$. Notice that 
$H[\MS_{0:\MOrder} | \CausalState_\MOrder]$ not vanishing
gives a way to understand how $\PC(L)$ deviates from linear growth.
That is, the nonlinearity of the approach of $\PC(L)$ to $\PC$ is exactly the
coentropy $H[\MS_{0:L} | \CausalState_L]$.

This property is tantamount to a very simple test to determine if a process is
a spin chain. If one obtains a plot similar to
Fig.~\ref{fig:EntropyGrowthCurvesCryptic} for the process in question, it is a
spin chain if $H[\MS_{0:L}, \CausalState_L]$ goes from $(0, \Cmu)$ flat to
$(\MOrder, \Cmu)$, and then follows $\EE + L \hmu$. That is, the block-state
entropy curve is flat until it reaches its asymptote at $L = \COrder = \MOrder$,
at which point it tracks it.

Furthermore, given (i) the above proof, (ii) the concavity proof from
Sec.~\ref{sec:bse}, and (iii) the fact that $k \le \MOrder$, for a given $\EE$,
$\hmu$, and $\MOrder$ spin chains are seen to be maximally-cryptic 
processes. By this we mean that for all processes with a particular set of 
values for $\EE$, $\hmu$, and $\MOrder$, the process that maximizes $\PC$ is 
a spin chain. This implies that $\Cmu$ is also maximized.

\begin{figure}[h!]
  \centering
  \includegraphics{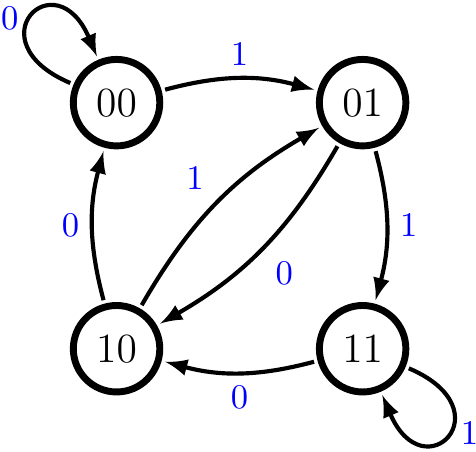}
  \caption{An order-2 Markov spin chain with full support.}
  \label{fig:spinchain}
\end{figure}

\begin{figure}[h!]
  \centering
  \includegraphics{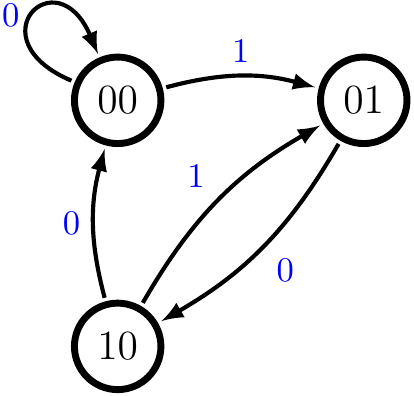}
  \caption{An order-2 Markov spin chain with partial support.}
  \label{fig:dimspinchain}
\end{figure}

\begin{figure}
  \centering
  \includegraphics{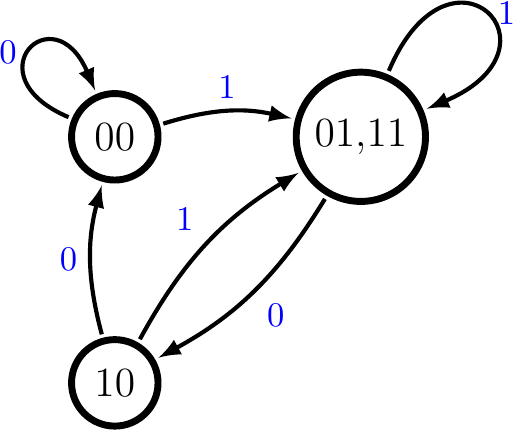}
  \caption{An order-2 Markov process, but not a spin chain.}
  \label{fig:nonspinchain}
\end{figure}

Figures~\ref{fig:spinchain} and \ref{fig:dimspinchain} show two order-$2$ Markov
spin chains. The first is a full-support order-$2$ Markov chain, while the
second has only partial support. In fact, the latter process has the Golden
Mean support consisting of all bi-infinite sequences that do not
contain consecutive $0$s.

Figure~\ref{fig:nonspinchain} gives an \eM\ of similar structure to the spin
chains just examined and, while it is also an order-$2$ Markov process, it is
not a spin chain. The reason is that one causal state (labeled ``$01,11$'') is
induced by two words: $01$ and $11$. This means that the correspondence between 
inducing-words and causal states is broken. It is no longer a spin chain.

We close this section with a number of open questions about spin chains. The
first two regard the structure of spin chains. If an \eM\ is a subgraph of an
order-$\MOrder$ Markov skeleton, then is it a spin chain? That is, does the
removal of an edge from a spin chain produce another spin chain? The intuition
behind this question is straightforward: Removing transitions disallows blocks,
but it would not cause any block to be associated with a different state. A
related question asks if all spin chains are of this form.

The next two questions regard the transformation from a spin chain to any other
process and vice versa. First, can any order-$\MOrder$ Markov, order-$\COrder$
cryptic \eM\ be obtained by starting with an order-$\MOrder$ Markov skeleton,
reducing some probabilities to zero and adjusting others to cause state 
merging? Also, given an order-$\MOrder$ Markov, order-$\COrder$ cryptic \eM, 
we can break the existing degeneracy so that $H[\MS_{0:\MOrder}] = \Cmu$. How 
does the nonspin chain we started with compare with the spin chain we end up with?

\section{Geometric Constraints}
\label{sec:geometric}

The geometry of the block entropy convergence illustrated in
Fig.~\ref{fig:EntropyGrowthCurvesCryptic} can be exploited. In
particular, as we will now show, a variety of constraints leads to further
results on the allowed convergence behaviors the block and block-state entropy
curves can express. Figure \ref{fig:BlockAndBlockStateEntropy}
depicts these results graphically.

\begin{figure}[h!]
  \centering
  \includegraphics[width=\linewidth]{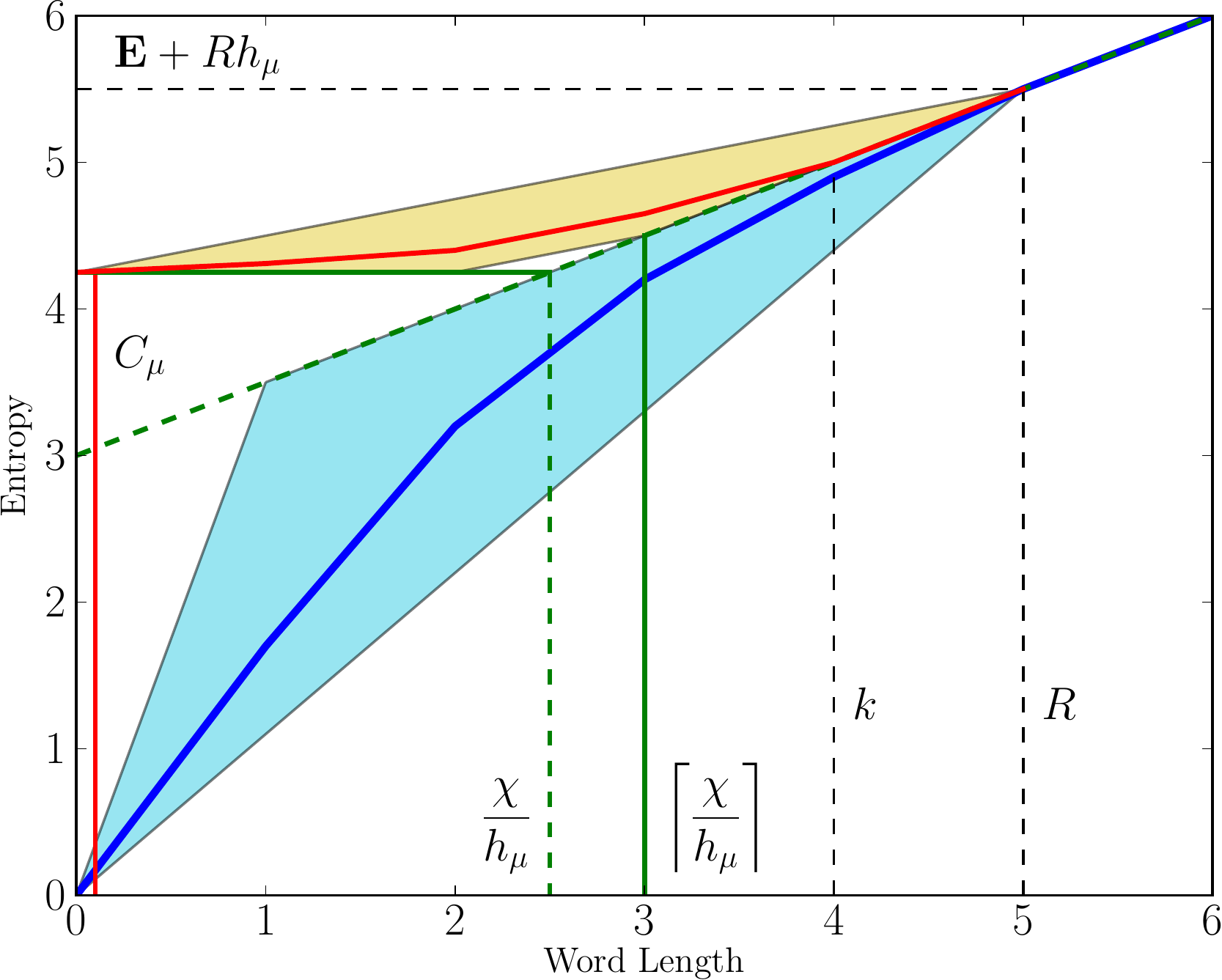}
\caption{Constraints on entropy convergence, illustrated for a process that
  is order-$5$ Markov and order-$4$ cryptic. The blue region circumscribes
  where the block entropy curve can lie; the tan, where the block-state
  entropy may be. These and the discreteness of $L$ lead to restrictions
  on allowed cryptic orders as well.
  }
\label{fig:BlockAndBlockStateEntropy}
\end{figure}

First, given the block entropy's concavity and it's asymptote, one sees that
the block entropy curve is contained within the triangle described by
$\{ (0, 0), (0, \EE), (R, H[\MS_{0:R}]) \}$. We also know that the block
entropy cannot grow faster than $H[\MS_0]$ and this excludes the triangle
$\{ (0, 0), (0, \EE), (1, H[\MS_0]) \}$. The resulting allowed region is
shown in light blue in Fig. \ref{fig:BlockAndBlockStateEntropy}.

Second and similarly, the block-state entropy's own properties require it to
be within a triangle described by
$\{ (0, \Cmu), (\frac{\PC}{\hmu}, \Cmu), (R, H[\MS_{0:R}]) \}$.

Third, since the entropy functions are defined for discrete values of word
length $L$, we can go a little further than these observations. The
block-state entropy cannot intersect the asymptote $\EE + \hmu L$ at a
noninteger $L$. Therefore the small triangle
$\{ (\lfloor\frac{\PC}{\hmu}\rfloor, \Cmu), (\frac{\PC}{\hmu}, \Cmu),
(\lceil\frac{\PC}{\hmu}\rceil, \EE + \lceil\frac{\PC}{\hmu}\rceil \hmu) \}$
is excluded. The resulting allowed trapezoid is displayed in tan in
Fig. \ref{fig:BlockAndBlockStateEntropy}.

Fourth, recalling results on the block-state entropy, this exclusion means that
processes with $\Cmu \neq \EE + \hmu k$, for some $k$, must have a degree of
nonoptimal retrodiction. In short, they are prevented from being spin chains.

Finally, given a process that has cryptic order $k$, we see that
$\Cmu \leq \EE + \hmu k$. A more detailed result then says that
$\Cmu = \EE + \hmu k$ if and only if $H[\MS^k] = \Cmu$. Moreover, it is
Markov order-$k$; that is, it is a spin chain.

\begin{figure*}
  \centering
  \includegraphics[width=1.0\textwidth]{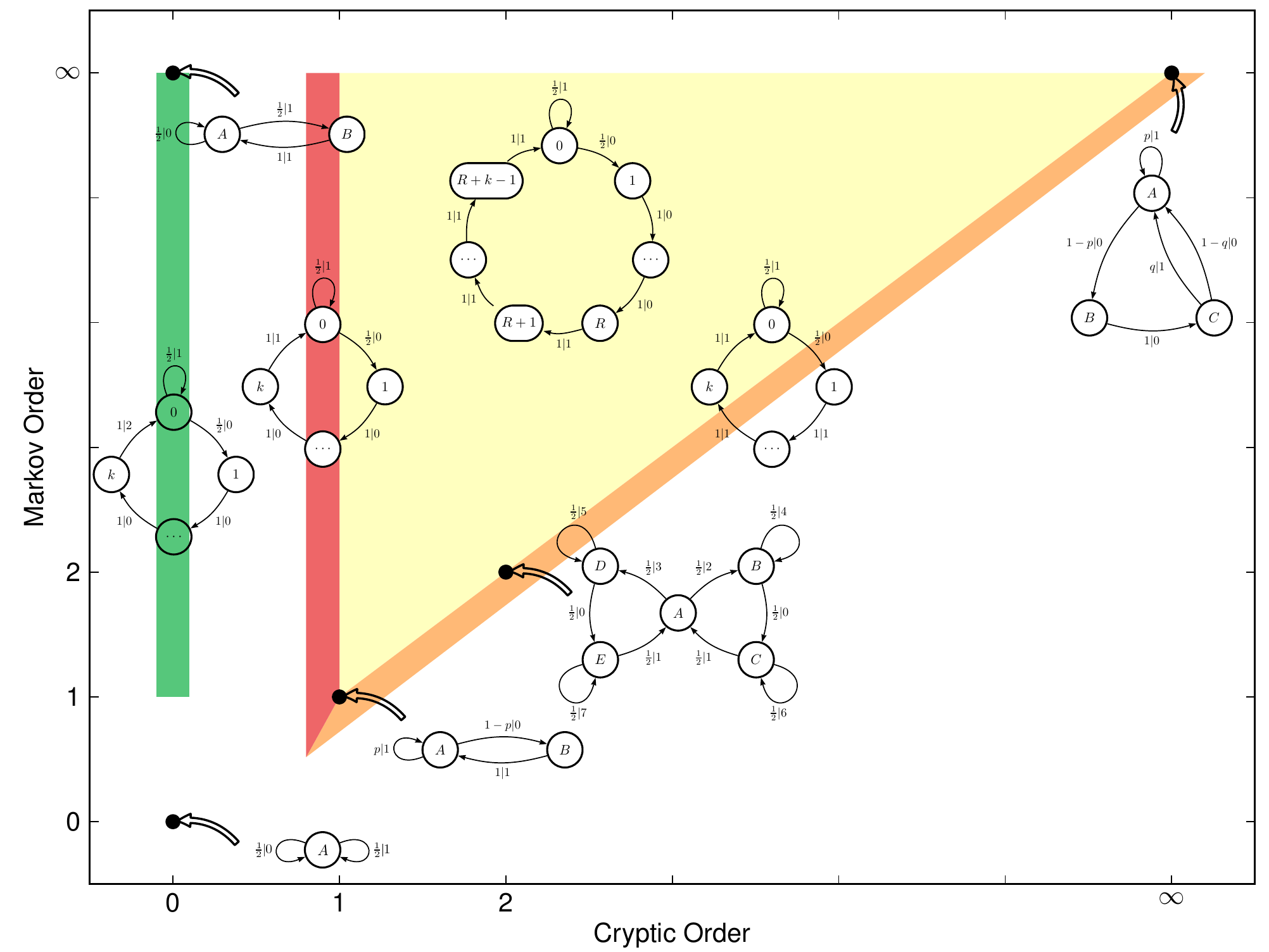}
\caption{The crypticity-Markovity roadmap for finite-state stationary processes:
  The range of possible
  Markov and cryptic orders, illustrated by a sample of processes depicted
  by their \eMs. Lower left: The Fair Coin Process and all other IID processes.
  Upper left: The $\infty$-cryptic Even Process. Upper right: The Nemo Process.
  Left vertical (green) line: The co-unifilar processes.
  }
\label{fig:CrypticMarkovRoadmap}
\end{figure*}

\section{The Cryptic Markovian Zoo}
\label{sec:process-zoo}

It turns out that there exist finite-state processes with all combinations of 
Markov and cryptic order; subject, of course, to the constraint that
$R \geq k$. These range from the zero structural complexity independent,
identically distributed processes, for which $R = 0$ and $k = 0$, to
few-state processes where either or both are infinite. (For a complementary
and exhaustive survey see Ref. \cite{Jame10a}.) In practice, given what we
now know about these properties, it is not difficult to design a variety of
processes that fulfill a given specification.

Also noteworthy is how the introduction of the new crypticity ``coordinate''
affects our view of several well studied examples. For instance, the Even
Process \cite{Crut01a} is one of the canonical finite-state, infinite-order
Markov processes. In the past, it was often thought of as representing
both intractability and compactness. Now, though, we see that it is trivial,
being $0$-cryptic. The Golden Mean Process, one of the simplest (order-$1$ 
Markov) subshifts of finite-type studied is now seen as more sophisticated, 
being $1$-cryptic. These and similar explorations naturally lead one to delve
deeper to find extreme examples---such as the Nemo process below---that are infinite
in both cryptic and Markov orders. Again, see Ref. \cite{Jame10a}.

Figure~\ref{fig:CrypticMarkovRoadmap} presents a crypticity-Markovity roadmap for
the space of finite-state processes. Borrowing from the immediately preceding
citations, it also displays a select few processes using their \eMs\ to show
concretely the full diversity of possible Markov and cryptic orders a
finite-state process can possess. The
green bar at $k=0$ consists of all co-unifilar processes. The orange line
contains all processes where the Markov and cryptic orders are identical---a
subset of which are the spin chains. All other processes lie above this line.
The Even Process is in the upper left corner. The Golden Mean
Process (no consecutive $0$s) is in the lower left. The $\infty$-cryptic,
infinite-order Markov Nemo Process is in the upper right corner. Several of
the other prototype \eMs\ depicted illustrate $(R,k)$-parametrized
classes of process for whom the Markov and cryptic orders can be
selected arbitrarily.

\section{Information Diagrams for Stationary Processes}
\label{sec:IDiagrams}

\emph{Information diagrams}, or simply \emph{I-diagrams}, are an important
analysis tool in using information theory to analyze multivariate stochastic
processes \cite{Yeun91a}. They are particularly useful when working with
processes and, as we have already seen here, give a good deal of insight when 
the \eM\ presentation is
employed \cite{Crut08a,Crut08b}.

The essential idea is that there is a one-to-one correspondence between
information-theoretic quantities---mutual information and conditional and
joint entropies---and measurable sets. Constructively, informational
relationships and constraints are depicted via set-theoretic operations:
joint entropies are set unions, conditional entropies correspond to set 
difference, mutual information corresponds to set intersection, and the 
like. The mathematical structure is a sigma
algebra over the process's events (words). The noncomposite sets are the \emph{atoms}
of the sigma algebra and their size is the magnitude of the corresponding
informational quantities. When depicted graphically, though, one often ignores
magnitudes and, instead, focuses on the set-theoretic relationships.

Armed with simple and familiar rules, one can often accomplish several
algebraic calculational steps on compound entropy expressions via
a simple I-diagram and a small
description. Perhaps more importantly, I-diagrams afford a visual calculus
that lends a heightened intuition about complicated relationships among
random variables.

Figures \ref{fig:FoliationMarkov} through \ref{fig:SpinChainIDiagram} show how
to make more explicit and intuitive the preceding formal views of entropy
convergence and its relationship to Markovity and crypticity. The two large
circles in each represent the past via $H[\Past]$ and the future via
$H[\Future]$. The excess entropy $\EE = I[\Past;\Future]$, being a mutual
information, is their intersection. The I-diagrams there show the nested
dependence of the various information measures as one increases block size
and so increases the number of random variables. In the general
multivariate case this would lead to an explosion of atoms.
However, due to the nature of processes and the \eM\ itself many
simplifications are possible.
Figures~\ref{fig:FoliationCrypticReqk}-\ref{fig:FoliationCrypticRgegek}
also depict the \eM's causal-state information, $\Cmu = H[\CausalState]$, as
a circle entirely inside the past $H[\Past]$.
This is so, since the causal states are a function of the past.

\begin{figure}[h!]
\centering
\includegraphics[width=\linewidth]{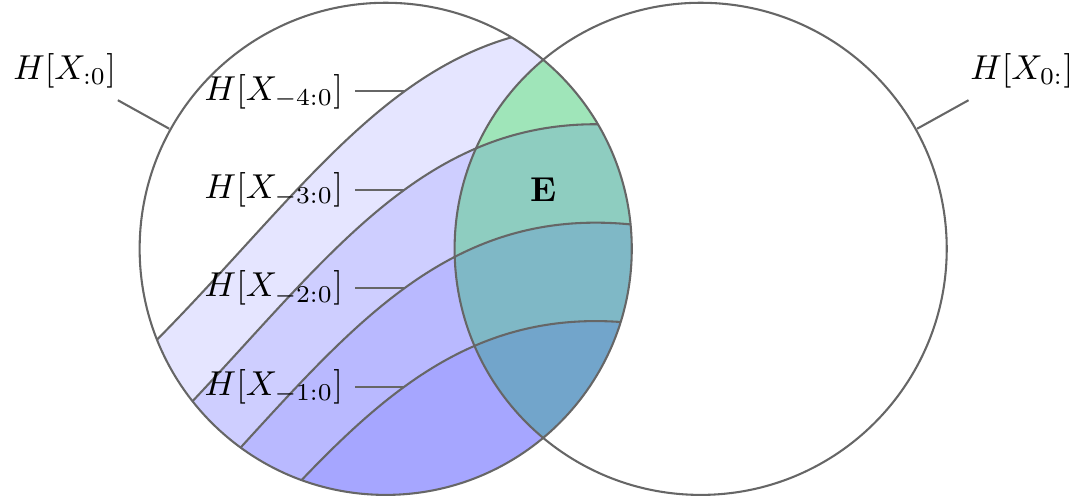}
\caption{Information diagram for an order-$4$ Markov process. Only the four
  most recent history symbols are needed to reduce as much uncertainty in
  the future as using the whole past would.
  }
\label{fig:FoliationMarkov}
\end{figure}

To start with the simplest case, Fig.~\ref{fig:FoliationMarkov} gives the 
I-diagram for an order-$4$ Markov
process. As one expects, only the four most recent history symbols are needed 
to reduce as much uncertainty in the future as using the whole past would.
Equivalently, as soon as the history contains four symbols, all of the shared
information between the past and the future (the excess entropy $\EE$) is captured.

\begin{figure}[h!]
\centering
\includegraphics[width=\linewidth]{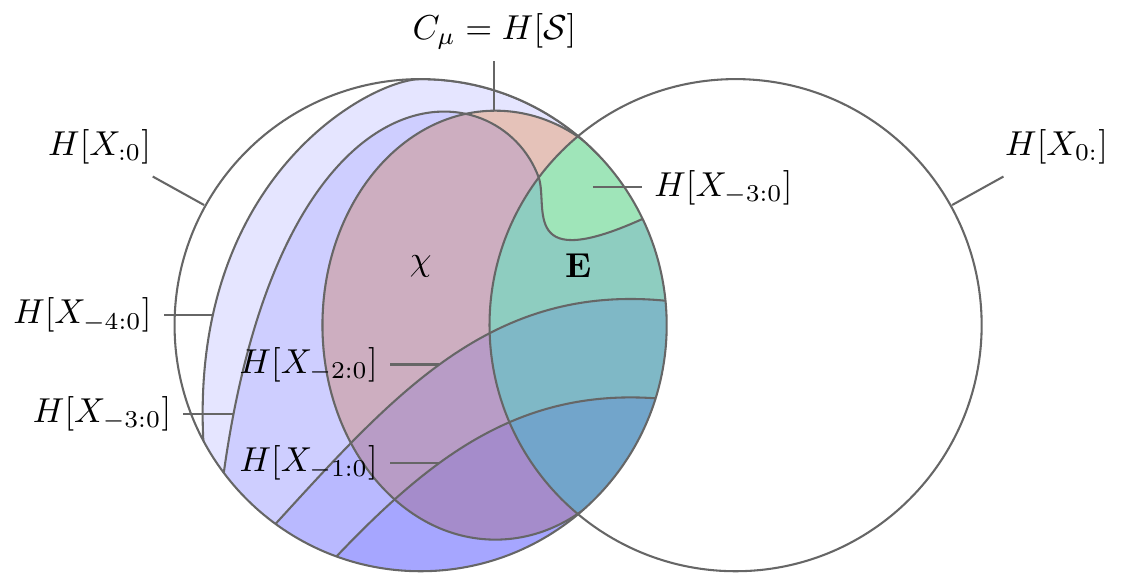}
\caption{Causal state is overlaid onto an I-diagram for an order-$4$ Markov
  process. As drawn, no fewer than $4$ history symbols are required to
  determine the causal state. The causal state, though, does not generally
  determine this length-four history.
  }
\label{fig:FoliationCrypticReqk}
\end{figure}

Figure \ref{fig:FoliationCrypticReqk} then overlays the causal state measure 
$H[\CausalState]$. In this, we see that no fewer than four history symbols are
required to determine the causal state. Importantly, it is now also made
explicit that causal states do not generally determine this history.

\begin{figure}[h!]
\centering
\includegraphics[width=\linewidth]{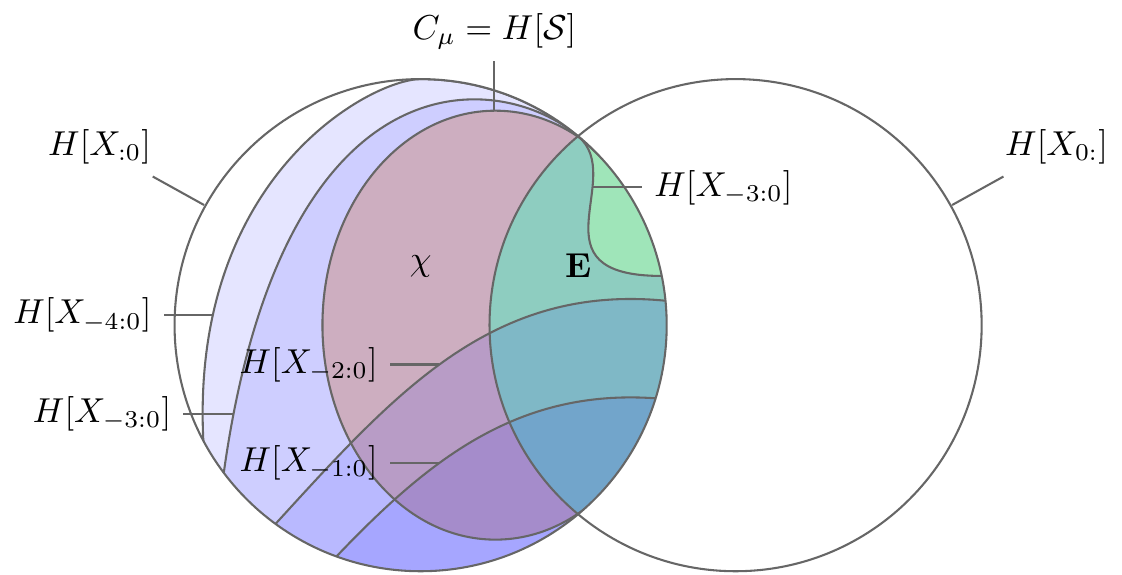}
\caption{An I-diagram for an order-$4$ Markov process, but order-$3$ cryptic.
  Four history symbols are required to determine the state, but only three
  are required if one conditions on the future.
  }
\label{fig:FoliationCryptic}
\end{figure}

Consider now the order-$4$ Markov, order-$3$ cryptic process of Figure
\ref{fig:FoliationCryptic}. As before, four history symbols are required
to determine the state. But, as depicted, only three history symbols are
required if one conditions on the future as well.

\begin{figure}[h!]
\centering
\includegraphics[width=\linewidth]{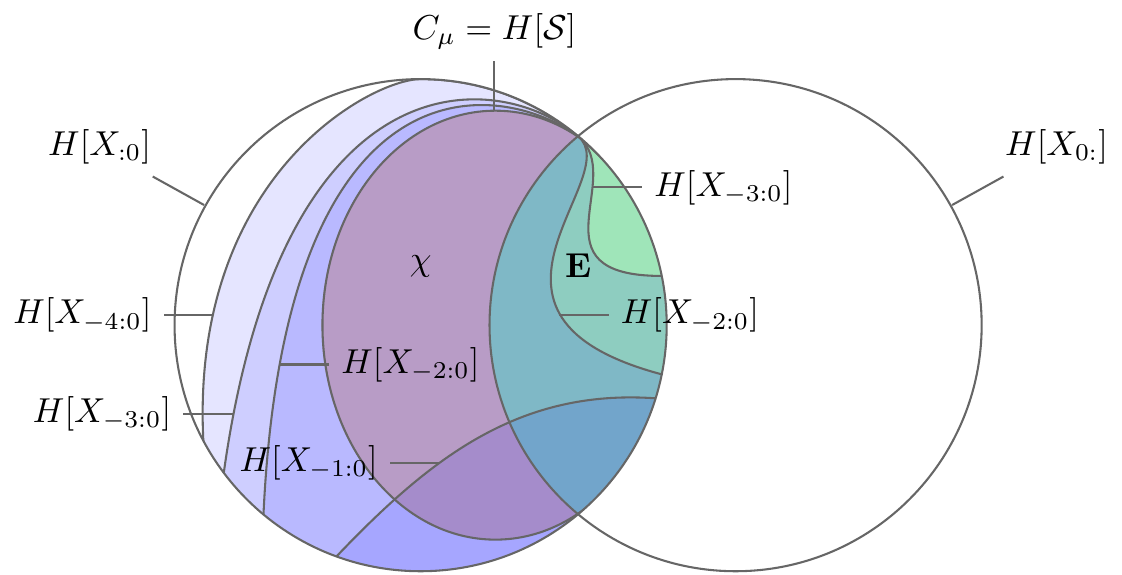}
\caption{The separation between Markov and cryptic orders can be widened:
  A Markov order-$4$, cryptic order-$2$ process.
  }
\label{fig:FoliationCrypticRgegek}
\end{figure}

Figure \ref{fig:FoliationCrypticRgegek} demonstrates how the
difference between Markov and cryptic orders can be increased
without bound. The I-diagram illustrates the sigma-algebra for an
order-$4$ Markov, order-$2$ cryptic process.

\begin{figure}[h!]
\centering
\includegraphics[width=\linewidth]{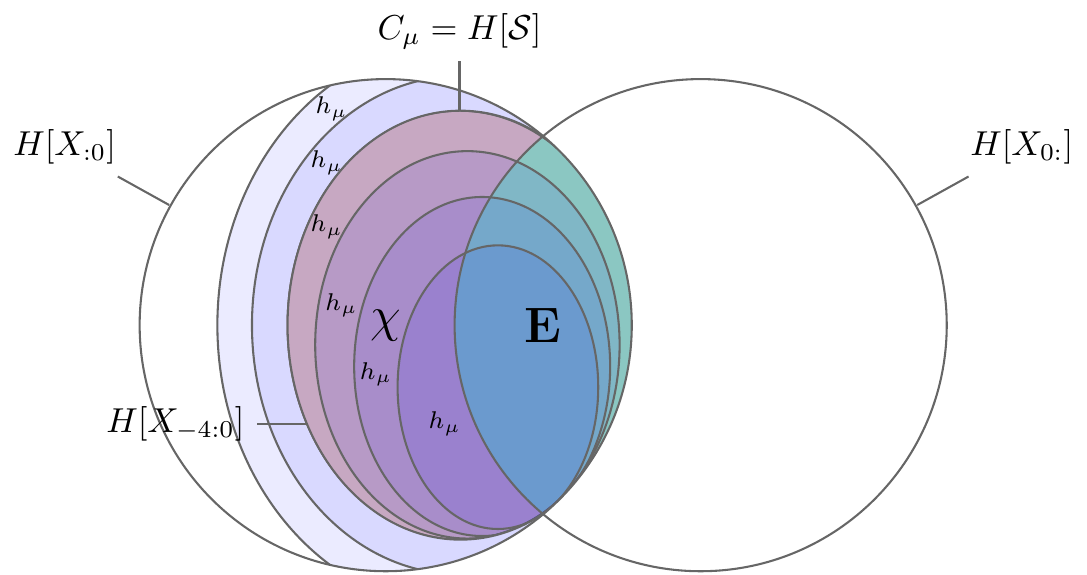}
\caption{The highly regular I-diagram for an order-$4$ spin chain.
  }
\label{fig:SpinChainIDiagram}
\end{figure}

Finally, Fig. \ref{fig:SpinChainIDiagram} gives the I-diagram for an order-$4$
spin chain. Several features of spin chains are clearly rendered in this
I-diagram. First, the shortest history that uniquely determines the state
occurs at length $4$. Specifically, as depicted,
$\min_L : H[\CausalState_0 | \MS_{-L:0}] = 4$. And, at the same time, this
length-$4$ history is itself uniquely determined by the causal state.

\section{Conclusion}

Crypticity, as the difference between a process's stored information and its
observed information, is a key property. The fundamental definitions, 
Eqs.~(\ref{def:Crypticity}) and (\ref{def:CrypticOrder}), though, are not
immediately transparent. However, they do lead to several
interpretations that prove useful in different settings. Given this, our
main goals were to explicate the basic notions behind crypticity and to
motivate various of its interpretations. Along the way, we provided a new
geometric interpretation for cryptic order, established a number of previously
outstanding properties, and illustrated crypticity by giving
a complete analysis for spin chains.

More specifically, using state-paths, we introduced several new interpretations 
of crypticity
that not only helped to explain the basic idea but also suggest future
applications in distributed dynamical systems. We also gave a simple
geometric picture that relates cryptic and Markov orders. We established
the equivalence between co-unifilarity and being $0$-cryptic, as well as the
concavity of the block-state entropy $H[\MS_{0:L} \CausalState_L]$.
We derived several geometric constraints and drew
out their implications for bounds on crypticity. These also led to an
improved bound on Markov order. Presumably, the bounds will help
improve estimates of crypticity and cryptic order, in both the finite
and infinite cases.

To give a sense of the relationship between cryptic and Markov orders we
gave a graphical overview classifying processes in their terms.
In a complementary way, we introduced the technique of foliated
information diagrams to analyze entropy convergence and Markov and
cryptic orders in terms of Shannon information measures and their
now block-length-dependent sigma algebra.

To ground the results in a concrete and familiar class of processes we analyzed
range-$R$ 1D spin chains in detail. We established their Markov order and showed that
the block-state entropy $H[\MS_{0:L} \CausalState_L]$ is flat for spin chains and that
$\PC(L) = L \hmu$, for all $L \leq R$. From these properties one can
determine whether or not a given process is representable as a spin chain:
Is the $R$-block entropy equal to the statistical complexity? The properties
also suggest what the processes in the neighborhood of a spin chain
look like.

Finally, by way of making contact with applications to physics and computation,
we close by briefly outlining the relationship between crypticity and
dynamical irreversibility in physical processes \cite{Elli11a}.
Consider the \emph{morph map} $\phi: \CausalState_0 \to \{ \MS_{0:} \}$.
A process's entropy rate controls the prediction uncertainty of this map:
$\hmu \equiv \lim_{L \to \infty} H[\MS_{0:L} | \CausalState_0]$. Now,
consider the state uncertainty determined by the inverse of the morph map:
$\phi^{-1}: \Future \to \{ \CausalState_0 \}$. This is already familiar.
The crypticity controls this uncertainty:
$\PC = \lim_{L \to \infty} H[\CausalState_0|\MS_{0:L}]$. Just as
the entropy rate is a process's rate of producing information, the crypticity
is its rate of information loss or, what one can call, a process's 
\emph{information-processing irreversibility}. And the latter, appropriately
adapting Landauer's Principle \cite{Land89a}, provides a lower bound on the
energy dissipation required to support a process's irreversible intrinsic computation.
We leave the full development of the thermodynamics of intrinsic
computation, however, to another venue.

\section*{Acknowledgments}

This work was partially supported by NSF Grant No. PHY-0748828 and by the
Defense Advanced Research Projects Agency (DARPA) Physical Intelligence
Subcontract No. 9060-000709. The views, opinions, and findings contained in this
article are those of the authors and should not be interpreted as representing
the official views or policies, either expressed or implied, of the DARPA or the
Department of Defense.

\appendix

\section{Why Crypticity?}
\label{sec:WhyCrypticity}

There are many ways to assemble information-theoretic quantities---more
specifically, information measures \cite{Yeun91a}. Why should one care
about crypticity and cryptic order? What makes them special?
We show that crypticity stands out among reasonable alternative
measures by a rather direct comparison.

It turns out that there are fewer information quantities than one
might expect---at least fewer interesting ones---over pasts,
futures, and states. Let's limit ourselves to quantities that
depend on only a finite set of objects and require that we look
for a ``1-parameter finitization'' property, based on block
length. In this case, we can make an exhaustive list of the
information measures and describe each one. The list, at first,
appears long. But this length is illustrative of the fact that
crypticity and cryptic order really do capture a relatively unique
process property. Everything else is either trivial, periodic, or Markov.

Table \ref{tab:Alternatives} presents the list.
It was assembled in a direct way by systematically writing down
alternative expressions over single variables, pairs of variables
and their joint and conditional entropy possibilities, over three
variables, and so on. One could also consider enumerating only the
relevant sigma-algebra atoms. This, however, obscures parallels to
existing quantities.

In addition, alternatives such as $H(\MS_{-L:0} | \Past)$ are not
included, since they are trivial. Nor were quantities such as
$H(\Past | \MS_{-L:0})$ added, although they could be. Quantities
along these lines would needlessly expand the list, to little benefit.

As elsewhere here, we assume the state random variable denotes a causal state.

\begin{table}[hbt]
\centering
\setlength{\extrarowheight}{0.07in}
\begin{tabular}{|c|c|}
\hline
Information Measure & Property Detected \\

\hline
$H[A]$\\
\hline
$H[\Past] = H[\MS_{-L:0}]$  & Periodic\\
$H[\Future] = H[\MS_{0:L}]$ & Periodic\\

\hline
$H[A|B]$\\
\hline
$H[\CausalState|\Past] = H[\CausalState|\MS_{-L:0}]$  & Markov\\
$H[\CausalState|\Future] = H[\CausalState|\MS_{0:L}]$ & Markov\\
$H[\Future|\Past] = H[\Future|\MS_{-L:0}]$ & Markov\\
$H[\Past|\Future] = H[\Past|\MS_{0:L}]$& Markov\\
$H[\Past|\CausalState] = H[\MS_{-L:0}|\CausalState]$  & Periodic\\
$H[\Future|\CausalState] = H[\MS_{0:L}|\CausalState]$ & Periodic\\

\hline
$H[A|BC]$\\
\hline
$H[\CausalState|\Past\Future] = H[\CausalState|\MS_{-L:0}\Future]$&Cryptic Order \\
$H[\CausalState|\Past\Future] = H[\CausalState|\Past\MS_{0:L}]$&Trivial \\
$H[\Past|\CausalState\Future] = H[\Past|\CausalState\MS_{0:L}]$&Trivial \\
$H[\Future|\Past \CausalState] = H[\Future|\MS_{-L:0}\CausalState]$&Trivial \\
$H[\Past|\CausalState\Future] = H[\MS_{-L:0}|\CausalState\Future]$&Periodic \\
$H[\Future|\Past \CausalState] = H[\MS_{0:L}|\Past\CausalState]$&Periodic \\

\hline
$H[AB]$\\
\hline
$H[\Past \CausalState] = H[\MS_{-L:0} \CausalState]$&Periodic \\
$H[\CausalState\Future] = H[\CausalState\MS_{0:L}]$&Periodic \\
$H[\Past\Future] = H[\MS_{-L:0}\Future]$&Periodic \\
$H[\Past\Future] = H[\Past\MS_{0:L}]$&Periodic \\

\hline
$H[AB|C]$\\
\hline
$H[\Past \CausalState|\Future] = H[\MS_{-L:0}\CausalState|\Future]$&Periodic \\
$H[\CausalState\Future|\Past] = H[\CausalState\MS_{0:L}|\Past]$&Periodic \\
$H[\Past \CausalState|\Future] = H[\Past\CausalState|\MS_{0:L}]$&Markov \\
$H[\CausalState\Future|\Past] = H[\CausalState\Future|\MS_{-L:0}]$&Markov \\

\hline
$H[ABC]$\\
\hline
$H[\Past \CausalState\Future] = H[\MS_{-L:0}\CausalState\Future]$&Periodic \\
$H[\Past \CausalState\Future] = H[\Past\CausalState\MS_{0:L}]$&Periodic \\

\hline
\end{tabular}
\caption{Alternative information measures over the past,
  the future, and the causal state, when they achieve their limit at
  finite block length $L$. As seen, almost all are either trivial,
  periodic, or detect the Markov property. Cryptic order stands
  out as unique.
  }
\label{tab:Alternatives}
\end{table}

\section{Equivalence of Forward and Reverse Restricted State-Paths}
\label{sec:RetroPath}

Why are the restricted state-paths the same in the forward and backward
lattice diagrams of Figs. \ref{fig:paths} and \ref{fig:paths_cryptic}?
Recall that a forward path is allowed if
$\Prob(\MS_{0:L} = w, \CausalState_1 =
\causalstate_B|\CausalState_0 = \causalstate_A) \neq 0$.
Similarly, a backward path is allowed when $\Prob(\CausalState_0 =
\causalstate_A, \MS_{0:L} = w |\CausalState_L = \causalstate_B)
\neq 0$. Since both causal states $\causalstate_A$ and
$\causalstate_B$ have nonzero probability by definition of being
recurrent, we see that we can state both cases as paths for which
$\Prob(\CausalState_0 = \causalstate_A, \MS_{0:L} = w,
\CausalState_L = \causalstate_B) \neq 0$.

\begin{figure}[ht]
\includegraphics{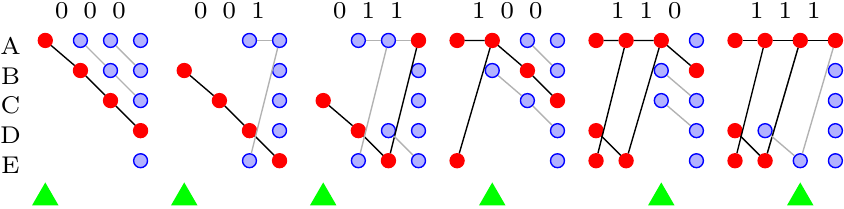}
\caption{Why forward and backward restricted paths are the same.
  In this figure state-paths are traced back from final states.
  Cf. Fig.~\ref{fig:paths_cryptic}.
  }
\label{fig:paths_retro}
\end{figure}

Figure \ref{fig:paths_retro} illustrates this by tracing state-paths backward
through the machine starting at each final state. Of course, since processes
and their \eMs\ are generally not counifilar, there will be splitting in these
paths. For
example, consider the paths that end in state $A$ on a $1$. $A$'s predecessors
on a $1$ are states $A$ and $E$.

Note that this produces a different initial set of candidate state-paths, when
compared with those in light blue in Fig.~\ref{fig:paths_cryptic}. Now, 
eliminate all paths that do not trace back successfully along the entire word.
Fig.~\ref{fig:paths_retro} shows these remaining state-paths in red and we see
that they are the same as those in Fig.~\ref{fig:paths_cryptic}.

\section{Crypticity and Co-unifilarity}
\label{app:counif}

Here, we explore the equivalence of $\EE = \Cmu$, co-unifilarity, and
$0$-crypticity using several results obtained in Ref.~\cite{Crut08b}. With a
small modification, the latter results allow for a more straightforward
proof that leads to a better understanding of these relations.

The ``forward'' argument is that $\PC(L) = 0$ implies crypticity
vanishes at all $L$. First, we recall two results.

Corollary 6 \cite{Crut08b}: If there exists a $k \ge 1$ for which $\PC(k)
= 0$, then $\PC(j) = 0$ for all $j \ge 1$.

Proposition 3 \cite{Crut08b}: $\lim_{k \to \infty} \PC(k) = \PC$.

Combining Cor. 6 and Prop. 3, we have the following:
If there exists a $k \ge 1$ for which $\PC(k)
= 0$, then $\PC(j) = 0$ for all $j \ge 1$ and $\PC = 0$.

The ``backward'' argument is that vanishing in the limit implies
crypticity vanishes at all $L$.

Since $\PC(k)$ is nonnegative (conditional entropy) and nondecreasing
(Prop. 2 \cite{Crut08b}) and limits to $\PC$ (Prop. 3 \cite{Crut08b}), we
have that $\PC = 0$ implies $\PC(k) = 0$, for all $k \ge 0$.

All that remains is to recall that co-unifilarity is identical to $\PC(1) = 0$
and this establishes the desired chain of implications:
\begin{align*}
\text{Co-unifilar}  & \iff \PC(1) = 0 \\
 & \iff \exists ~k \ge 1 : \PC(k) = 0 \\
 & \iff \PC(k) = 0, \forall ~k \ge 0 \\
 & \iff \PC = 0 \\
 & \iff 0\text{-cryptic} ~.
\end{align*}
The heart of the result falls in the middle. It shows us that any nontrivial
zero in $\PC(k)$ is equivalent to the entire function, as well as $\PC$
itself, vanishing.


\end{document}